\documentclass[12pt]{amsart}

\usepackage{fullpage}
\usepackage{amsmath, amsthm, amsfonts}
\usepackage{color}
\usepackage{parskip}
\usepackage{graphicx}

\usepackage{pdfpages}
\usepackage{tikz}
\usetikzlibrary{arrows}

\newtheorem{theorem}{Theorem}[section]
\newtheorem{corollary}[theorem]{Corollary}
\newtheorem{lemma}[theorem]{Lemma}

\title{The Spread of Voting Attitudes in Social Networks}

\date{}
\begin{document}
\begin{center}
\maketitle
	
	 Jordan Barrett$^{A}$ \quad \quad Christopher Duffy$^{B,}$\footnote{Corresponding Author: christopher.duffy@usask.ca}$^{,2}$ \quad \quad Richard Nowakowski$^{C,}\footnote{Research supported by the Natural Science and Engineering Research Council of Canada}$
	 
	 	\vspace{.15in}
	 \begin{small}
	 $^{A}$Department of Mathematics and Statistics, McGill University, Montreal, CANADA\\
	 $^{B}$Department of Mathematics and Statistics, University of Saskatchewan, Saskatoon, CANADA\\
	 $^{C}$Department of Mathematics and Statistics,  Dalhousie University, Halifax, CANADA\\
	\end{small}

\end{center}

\begin{abstract}
	The Shapley-Shubik power index is a measure of each voters power in the passage or failure of a vote. We extend this measure to graphs and consider a discrete-time process in which voters may change their vote based on the outcome of the previous vote. 
	We use this model to study how voter influence can spread through a network.
	We find conditions under which a vanishingly small portion of consenting voters can change the votes of the entirety of the network.
	For a particular family of graphs, this process can be modelled using cellular automata. 
	In particular, we find a connection between this process and the well-studied cellular automata, Rule 90.
	We use this connection to show that such processes can exhibit arbitrarily-long periodicity.
\end{abstract}


\section{Introduction and Definitions}

Network effects heavily impact the spread of ideas and opinions among a population.
And so discrete-time processes on graphs provide an ideal tool for studying such a spread.
A wide variety of updating mechanisms, both deterministic and stochastic, have been proposed to model the evolution of attitudes of vertices in a graph \cite{A14,B07,D83}.
Particular attention has been paid to such models that consider the outcome of local voting procedures as part of the updating mechanic \cite{C18,C16,C89,H01}, usually with the aim of computing the expected time to consensus.
In addition to its possible sociological applications, these evolving systems are important topic in distributed computing \cite{P02}.
In this work we study a fundamental tool in the study of voting systems as an update mechanism for such a discrete-time process. 

For votes passing by a simple majority, one can intuit that votes that pass by a slim majority feature consenting voters that have  more power than consenting voters in votes that pass by a wide majority.
The Shapley-Shubik power index is a measure of each voters power in the passage or failure of a vote.
We extend this measure of voter power to graphs and use it to examine the discrete-time spread of voter attitudes in a network.

Shapley and Shubik introduced their index as a first examination of the problem of developing and maintaining a legislative body \cite{S54}. 
They noted that revisions to the structure of a legislative body may include new forms of bias, unintended by the revisers. 
Their mathematical evaluation of the division of power within a legislative body is a tool to examine overall fairness in the system, where fairness can include properties such as equal representation, and protection of minority interests.
Their tool has since been used to study bodies as varied as the Council of Ministers of the European Council \cite{W94}, the US Electoral College \cite{B01}, the Polish government \cite{M09}.

In a voting situation, the \textit{Shapley-Shubik power index} is applicable when votes are taken in order, a roll
call for example. 
A person, $I$, is \textit{pivotal} if before their vote the motion had not passed but when they vote the motion is carried. 
Formally, if there are $n$ voters and $\pi_I$ is the number of permutations in which $I$ is pivotal then the
\textit{power} of $I$ is $p(I)=\pi_I/n!$. 

To extend the Shapley-Shubik power index to a discrete-time process on a graph we require the following definitions. 
Let $G$ be a graph. 
A \emph{configuration} of $G$, $C: V(G) \to \{C,D\}$ assigns each vertex to be either a collaborator $(C)$ or a defector $(D)$. 
We refer to $C(v)$ as the \emph{strategy of $v$}. 
If $C$ is a configuration and $v \in V(G)$ the \emph{collaborator neighbourhood} of $v$, denoted $N_C[v]$ is the set of vertices in the closed neighbourhood of $v$ that are collaborators. 
We define the \emph{defector neighbourhood}, denoted $N_D[v]$, analogously.

Let $C$ be a configuration of $G$, $v \in V(G)$ and $w \in \left[\frac{1}{2},1\right)$. 
If $v$ is a collaborator (i.e, if $C(v) = C$), then the \emph{power of $v$}  is given by $p(v)= \frac{1}{|N_C[v]|}$ when $\frac{|N_C[v]|}{|N[v]|} > w$, otherwise it is $0$. 
If $v$ is a defector then $p(v) = \frac{1}{|N_D[v]|}$ when $\frac{|N_C[v]|}{|N[v]|} \leq w$, and $0$ otherwise. 
We may interpret  $w$ as the proportion of dissenting members needed to defeat the vote. 
We refer to $w$ as the \emph{win condition}.
Notice that in each case, voters on the losing side have no power.
For standard graph theoretic notations we refer the reader to \cite{Bondy08}.

Consider the following example of five vertices where each vote passes by a simple majority (i..e, $w = 1/2$). 
In Figure \ref{fig:example1}, and in subsequent figures, black vertices denote collaborators and white vertices denote defectors. 

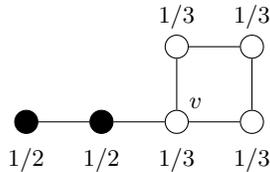
\begin{figure}
\[
\begin{scriptsize}
\begin{tikzpicture}
[scale = 1]
\tikzset{vertex/.style = {shape=circle,draw}}

\node[vertex,fill=black] (00) at (0,0) {};
\node[vertex,fill=black] (10) at (1,0) {};
\node[vertex] (20) at (2,0) {};
\node[vertex] (30) at (3,0) {};
\node[vertex] (21) at (2,1) {};
\node[vertex] (31) at (3,1) {};

\node (L00) at (0,-0.5) {$1/2$};
\node (L10) at (1,-0.5) {$1/2$};
\node (L20) at (2,-0.5) {$1/3$};
\node (L30) at (3,-0.5) {$1/3$};
\node (L21) at (2,1.4) {$1/3$};
\node (L31) at (3,1.4) {$1/3$};

\node (v20) at (2.25,0.25) {$v$};

\draw (00) -- (10) -- (20) -- (30) -- (31) -- (21) -- (20);

\end{tikzpicture}
\end{scriptsize}
\]
	\caption{Computing the power of a vertex}
	\label{fig:example1}	
\end{figure}

Vertex $v$ is a collaborator. If a simple majority vote was taken across the vertices in $v$'s closed neighbourhood, then the vote would pass with $3$ votes for and $1$ vote against.
Therefore $p(v) = \frac{1}{3}$.

For a fixed graph $G$, a configuration $C_0$ and  $w \in \left[\frac{1}{2},1\right)$ the $w$-power index process is a discrete-time process defined as follows.
In each round $t > 0$, each vertex takes the strategy of the vertex in its closed neighbourhood that had the greatest power at the end of the previous round. 
If this strategy is not well-defined, i.e., if there are two vertices with the greatest power and differing strategies,  then the vertex retains its strategy. 
Each vertex performs this update simultaneously.  
We refer to the configuration at the end of round $i$ as $C_i$ and to $C_0$ as the \emph{initial configuration}. 
Taking the configuration in Figure \ref{fig:example1} as $C_0$, Figure~\ref{fig:example2}  gives the configuration $C_1$ and the resulting powers.
Notice that in $C_0$ vertex $v$ had $p(v) = \frac{1}{3}$.
This vertex has a defector neighbour with $p = \frac{1}{2}$, and so $v$ has changed their strategy from collaboration to defection.

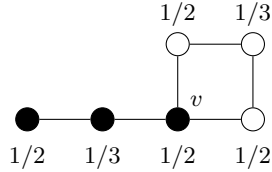
\begin{figure}[h]
\[
\begin{scriptsize}
\begin{tikzpicture}
[scale = 1]
\tikzset{vertex/.style = {shape=circle,draw}}

\node[vertex,fill=black] (00) at (0,0) {};
\node[vertex,fill=black] (10) at (1,0) {};
\node[vertex,fill=black] (20) at (2,0) {};
\node[vertex] (30) at (3,0) {};
\node[vertex] (21) at (2,1) {};
\node[vertex] (31) at (3,1) {};

\node (L00) at (0,-0.5) {$1/2$};
\node (L10) at (1,-0.5) {$1/3$};
\node (L20) at (2,-0.5) {$1/2$};
\node (L30) at (3,-0.5) {$1/2$};
\node (L21) at (2,1.4) {$1/2$};
\node (L31) at (3,1.4) {$1/3$};

\node (v20) at (2.25,0.25) {$v$};

\draw (00) -- (10) -- (20) -- (30) -- (31) -- (21) -- (20);

\end{tikzpicture}
\end{scriptsize}
\]
	\caption{The process after a single time-step}	
	\label{fig:example2}

\end{figure}

We continue with a second example to highlight the cyclic nature of this process.
Figure \ref{fig:cyclExample} gives the evolution of the process for the given initial configuration.

\begin{figure}
\[
\begin{scriptsize}
\begin{tikzpicture}
[scale = 1]
\tikzset{vertex/.style = {shape=circle,draw}}

\node[vertex] (00) at (0,-1) {};
\node[vertex] (01) at (0,0) {};
\node[vertex] (105) at (1,-0.5) {};
\node[vertex,fill=black] (20) at (2,-1) {};
\node[vertex,fill=black] (21) at (2,0) {};

\draw (00) -- (01) -- (105) -- (00);
\draw (20) -- (21) -- (105) -- (20);

\end{tikzpicture}
\end{scriptsize}
\ \ \ \ \phantom{\rightarrow} \ \ \ \ \ \ \ 
\begin{scriptsize}
\begin{tikzpicture}
[scale = 1]
\tikzset{vertex/.style = {shape=circle,draw}}

\node[vertex] (00) at (0,0) {};
\node[vertex] (01) at (0,1) {};
\node[vertex,fill=black] (105) at (1,0.5) {};
\node[vertex,fill=black] (20) at (2,0) {};
\node[vertex,fill=black] (21) at (2,1) {};

\draw (00) -- (01) -- (105) -- (00);
\draw (20) -- (21) -- (105) -- (20);

\end{tikzpicture}
\end{scriptsize}
\ \ \ \ \phantom{\rightarrow} \ \ \ \ \ \ \ 
\begin{scriptsize}
\begin{tikzpicture}
[scale = 1]
\tikzset{vertex/.style = {shape=circle,draw}}

\node[vertex] (00) at (0,0) {};
\node[vertex] (01) at (0,1) {};
\node[vertex] (105) at (1,0.5) {};
\node[vertex,fill=black] (20) at (2,0) {};
\node[vertex,fill=black] (21) at (2,1) {};

\draw (00) -- (01) -- (105) -- (00);
\draw (20) -- (21) -- (105) -- (20);

\end{tikzpicture}
\end{scriptsize}
\]
\[
C_0
\ \ \ \ \ \ \ \ \ \ \ \ \ \ \ \ \ \ \ \ \ \ \ \ \ \ \ \
C_1
\ \ \ \ \ \ \ \ \ \ \ \ \ \ \ \ \ \ \ \ \ \ \ \ \ \ \ \ \
C_2
\]
	\caption{A  $\frac{1}{2}$-power index process that is periodic with period $2$}	
	\label{fig:cyclExample}

\end{figure}
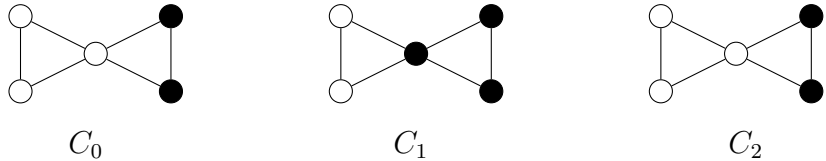

Here we notice $C_0 = C_2$.
As the process is deterministic, this implies $C_0 = C_2 = C_4 \dots$ and $C_1 = C_3 = C_5 \dots$.
For a fixed graph $G$, the evolution of the process depends only on $C_0$ and $w$.
With respect to a fixed $w\geq \frac{1}{2}$, for an initial configuration $C_0$ we say that the process becomes \emph{stable} if there exists $i\geq 0$ such that $C_i = C_{i+1}$. 
Alternatively the process becomes \emph{periodic with period $\ell$} if $\ell > 1$ is the least integer so that there exists $C_i$ where $C_i = C_{i + \ell}$ for some $i \geq 0$.
For the example in Figures \ref{fig:example1} and \ref{fig:example2}  the process becomes stable; as one can check $C_1 = C_2$.

For a fixed graph $G$, an initial configuration $C_0$ and a win condition $w \geq \frac{1}{2}$ the state of a vertex $v$ at time $t$ depends on the states of the vertices in the closed second neighbourhood of $v$. 
As such, we may consider the triple $(G,C_0,w)$ an as instance of a cellular automata on a graph with a particular rule set.
Following the preliminary results in Section \ref{sec:prelim}, we study the interplay between periodicity,  win condition and graph structure. 
For a particular class of cartesian products of graphs, we show that the evolution a particular triple $(G,C_0,w)$ follows exactly the evolution of a well-studied 1D cellular automaton, Rule 90 \cite{W84}.
This implies the existence of processes with arbitrarily long periods.
We further find that for every positive integer $k$ and any connected graph $H$, that there exists a graph $G$, of which $H$ is an induced subgraph, and a configuration $C_0$ such that the $\frac{1}{2}$-power index process on $C_0$ eventually cycles with cycle length at least $k$.

\section{Preliminary Results} \label{sec:prelim}
Of the configurations that become stable, we identify a particular class of those that stabilise such that every vertex has the same strategy -- the vertices reach a consensus.
If $C_0$ results in stable configuration $C_i$ such that each vertex in $C_i$ is a collaborator (defector), then we say that $C_0$ is \emph{collaborator dominant} (\emph{defector dominant}). 
Surprising, there exist collaborator (defector) dominant configurations such that $C_0$ contains relatively few collaborators (defectors).

For fixed  positive integers $j$ and $n$, denote by $G_{j,n}$ the graph formed from the disjoint union of $n+1$ cliques, $K_{j} \cup K_{2j} \cup K_{j2^2} \cup\dots \cup K_{j2^n}$, with edges between successive cliques such that 
\begin{itemize}
	\item each vertex of $K_{j}$ has exactly two neighbours in $K_{2j}$;
	\item each vertex of $K_{j2^i}$ has exactly two neighbours in $K_{j2^{i+1}}$ and one neighbour in $K_{j2^{i-1}}$, for $1 \leq i \leq n-1$; and
	\item each vertex of $K_{j2^n}$ has exactly one neighbour in $K_{j2^{n-1}}$.
\end{itemize}

Using $G_{j,n}$ we show that for any $w$, a vanishingly small proportion of collaborators or defectors can spread to fill the entire graph.

\begin{theorem} \label{prop:CMinWin}
	For $w < \frac{j}{j+2}$,  the $w$-power index process on $G_{j,n}$ with
	\[ C_0(v) = \begin{cases} 
	C & v \in K_j, \\
	D & \mbox{otherwise.} 
	\end{cases}
	\] is collaborator dominant.
	
\end{theorem}	

\begin{proof}
	Consider the $w$-power index process on $G_{j,n}$ with
	\[ C_0(v) = \begin{cases} 
	C & v \in K_j, \\
	D & \mbox{otherwise.} 
	\end{cases}\]

	We show for $1 \leq i \leq n+1$ we have
		\[ C_i(v) = \begin{cases} 
		C, & v \in K_j \cup K_{2j} \cup \dots \cup K_{j2^{i-1}} , \\
		D, & \mbox{otherwise.} 
		\end{cases} \]

	We consider first the case $i =1$.
	We compute

	\[ p_0(v) = \begin{cases} 
	\frac{1}{j}, & v \in K_j \\
	\frac{1}{2j+2} & v \in K_{2j} \\
	\frac{1}{j2^k + 3}, & v \in K_{j2^k}, 2 \leq k \leq n-1\\
	\frac{1}{j2^n + 1} & v \in K_{j2^n}
	\end{cases} \]
	
	Thus 
	\[ C_1(v) = \begin{cases} 
	C & v \in K_j \cup K_{2j}, \\
	D & \mbox{otherwise.} 
	\end{cases} \]
	
Consider now $2 \leq i \leq n-2$ so that
\[ 
C_i(v) = 
\begin{cases} 
C, & v \in K_j \cup K_{2j} \cup \dots \cup K_{j2^{i}} , \\
D, & \mbox{otherwise.} 
\end{cases} \] 

We compute
	\[ p_i(v) = 
\begin{cases} 
	\frac{1}{j+2}, & v \in K_j \\
	\frac{1}{j2^\ell +3} & v \in K_{j2^\ell}, 1 \leq \ell \leq i-1\\
	\frac{1}{j2^i+1} & v \in K_{j2^i}\\
	\frac{1}{j2^{\ell+1}+3} & v \in K_{j2^\ell}, i+1 \leq \ell \leq n-1\\
	\frac{1}{j2^n+1} & v \in K_{2^n}\\
\end{cases} \]

And so 
		\[ C_{i+1}(v) = \begin{cases} 
C, & v \in K_j \cup K_{2j} \cup \dots \cup K_{j2^{i+1}} , \\
D, & \mbox{otherwise.} 
\end{cases} \] 

Finally, consider

\[ C_{n-1}(v) = \begin{cases} 
C, & v \in K_j \cup K_{2j} \cup \dots \cup K_{j2^{n-1}} , \\
D, & \mbox{otherwise.} 
\end{cases} \] 

We compute
	\[ p_{n-1}(v) = 
\begin{cases} 
\frac{1}{j+2}, & v \in K_j \\
\frac{1}{j2^\ell +3} & v \in K_{j2^\ell }, 1 \leq \ell \leq n-2\\
\frac{1}{j2^{n-1}+1} & v \in K_{j2^{n-1}}\\
\frac{1}{j2^n} & v \in K_{2^n}\\
\end{cases} \]

And so 
$C_{n}(v) = C$ for all $v \in  V(G_{j,n})$. 
This completes the proof.
\end{proof}

\begin{theorem}
	For all $w$, the $w$-power index process on $G_{3,n}$ with
	
	\[ C_0(v) = \begin{cases} 
	D & v \in K_3, \\
	C & \mbox{otherwise.} 
	\end{cases}
	\] is defector dominant.
\end{theorem}	

\begin{proof}
	The proof of this result follows similarly to that of Theorem \ref{prop:CMinWin} and thus is omitted.
\end{proof}

\begin{corollary} \label{cor:CMinWin}
	For any $\epsilon > 0$, and any $w \in \left[\frac{1}{2},1\right)$ there exists a graph $G$ and a configuration $C_0$ such that the density of collaborators in $C_0$ is less than $\epsilon$ and $C_0$ is collaborator dominant.
\end{corollary}

\begin{proof}
	This follows directly from Proposition \ref{prop:CMinWin} with the observation  \[\frac{|V(K_j)|}{|V(G_{j,n})|} = \frac{j}{j2^{n+1}} \to 0 \mbox{ as } n \to \infty.\]
\end{proof}

\begin{corollary}\label{cor:DMinWin}
	For any $\epsilon > 0$, and any $w \in \left[\frac{1}{2},1\right)$ there exists a graph $G$ and a configuration $C_0$ such that the density of defectors in $C_0$ is less than $\epsilon$ and $C_0$ is defector dominant.
\end{corollary}

Corollaries \ref{cor:CMinWin} and \ref{cor:DMinWin} show that regardless the value of $w$, it is possible that relatively few collaborators or defectors can spread their influence to every vertex of the graph. 
In each case we notice that this happens after $diam\left(G_{j,n}\right)$ timesteps.

Next we examine the effect of the win condition, $w$, on the process. 
We show that for the fixed graph, the set of vertex valences give rise to an equivalence relation on $[\frac{1}{2}, 1)$ with respect to the behaviour of the process. 
That is, we partition $[\frac{1}{2}, 1)$ such that a pair of win conditions in the same part give rise to the same process.

Recall the example given in Figure \ref{fig:cyclExample}. For $w = \frac{1}{2}$ the process is cyclic with period $2$. However, it can be verified that if $w = 0.6 $, then the evolution of the process is unchanged.

Let $G$ be a $k$-regular graph, $v$  be a vertex of $G$ and $C$ a configuration such that $C(v) = C$.  
Since $v$ has degree $k$ it must be that $\frac{|N_C[v]|}{|N[v]|} \in \{\frac{1}{k+1}, \frac{2}{k+1}, \dots , \frac{k+1}{k+1} \}$. 
Since $G$ is $k$-regular, for any fixed value of $w \in [\frac{i}{k+1}, \frac{i+1}{k+1})$ the resulting behaviour of the process will be the same. 
By generalising to non-regular graphs, we arrive at the following definition.

For a graph $G$ and a vertex $v$ let $S_v = \left\{ \frac{i}{|N[v]|} :   \frac{ |N[v]|}{2} \leq i \leq |N[v]|, \; i \in \mathbb{Z} \right\}$ and $S_G~=~\bigcup_{v \in V(G)}~S_v$. 
The set $S_v$ consists of all non-zero possible values of $p(v)$.
And so the set $S_G$ is the set of non-zero possibilities for the power of a vertex in $G$.
The \emph{win partition} of $G$ is the partition   $\alpha = \{  [\frac{1}{2}, s_1),[s_1, s_2),[s_2,s_3)\dots [s_{|S_G|-1} ,s_{|S_G|}) \} $ of $\left[\frac{1}{2},1\right)$ where $s_i \in S_G$ and $s_1 < s_2  < \dots  < s_{|S_G|}$.

\begin{theorem}\label{thm:PartitionWorks}
	Let $G$ be a graph with win partition $ \alpha = \{\alpha_1, \alpha_2, \dots , \alpha_k\}$ and let  $w, w^\prime \in \alpha_i$.	
	If $C_1, C_2 \dots$ is the sequence of configurations resulting from the $w$-power index process with initial configuration $C_0$ and $C^\prime_1, C^\prime_2 \dots$ is the sequence of configurations resulting from the $w^\prime$-power index with initial configuration $C^\prime_0 = C_0$, then $C^\prime_i = C_i$ for all $i > 0$.
\end{theorem}

\begin{proof}
	Let $S_G = \{s_1, s_2, \dots s_\ell\}$  such that $s_1 < s_2 < \dots  <s_\ell$ and let $s_i < w <  w^\prime < s_{i+1}$. 
	We proceed by induction.
	
	Assume $C_i = C_i^\prime$. For all $v \in V(G)$ let $p_i(v)$ and $p^\prime_i(v)$ be the power index of $v$ in $C_i$ and $C_i^\prime$, respectively. If $deg(v) = k$, then $p_i(v), p^\prime_i(v) \in \{0,\frac{1}{k+1}, \frac{2}{k+1}, \dots, 1\}$. 
	We show  $p_i(v) = p^\prime_i(v)$.
	
	If $C_i(v) = C$, and $p_i(v) \neq 0$, then $\frac{|N_C[v]|}{k+1} > w$. If $p_i(v) \neq p^\prime_i(v)$, then $p^\prime_i(v) = 0$. This implies  $s_i \leq w < \frac{|N_C[v]|}{k+1} \leq  w^\prime < s_{i+1}$. However, by definition of $\alpha$ there exists $s_k$ such that $\frac{|N_C[v]|}{k+1} = s_k$, a contradiction. 
	A similar argument holds when  $C_i(v) = D$.
	Therefore if $C_i = C_i^\prime$ then $C_{i+1} = C_{i+1}^\prime$. 
	The result follows as $C_0 = C_0^\prime$.
\end{proof}

Theorem \ref{thm:PartitionWorks} implies that in studying the $w$-power index process for a fixed graph $G$, we need only consider finitely many values of $w$ -- one from each equivalence class implied by the win partition. 
In practise we take the included lower bound of each part as the representative element.

\begin{theorem}\label{thm:LastTwo}
	Let $G$ be a graph with win partition $\alpha = \{\alpha_1, \alpha_2, \dots , \alpha_k\}$. If $w \in \alpha_{k-1} \cup \alpha_k$, then every configuration $C_0$ on $G$ is eventually stable in the $w$-power index process.
\end{theorem}
\begin{proof}
	Let $G$ be a graph with win partition $\alpha = \{\alpha_1, \alpha_2, \dots , \alpha_k\}$.
	We proceed in cases.
	
	\emph{Case I: $w \in \alpha_{k-1}$} 
	
	By Theorem \ref{thm:PartitionWorks} we need only consider the $\frac{\Delta-1}{\Delta+1}$-power index process on $G$. Let $C_0$ be a configuration and let $p_0$ be the power  of vertices of $G$ with respect to $C_0$. Consider $v$ such that $C_0(v) = C$.
	 When $w= \frac{\Delta-1}{\Delta+1}$ we have $p_0(v) > 0$ if and only if $N[v] = N_C[v]$ or if $N_C[v] = \Delta$.  
	 If $N_C[v] = \Delta$ then $p_0(v) = \frac{1}{\Delta}$.   
	 Consider $v$ such that  $C_0(v) = D$. 
	 We have $p_0(v) > 0$ if and only if $|N_D[v]| \geq 2$. 
	 Observe that if $p(v) > 0$, but $N_C[v] \neq \emptyset$, then  $p_0(v) \geq \frac{1}{\Delta}$.
	
	We claim that if $C_i(x) = C_i(y) = D$ for $xy \in E(G)$, then $C_{i+1}(x) = C_{i+1}(y) = D$. 
	Without loss of generality, assume $C_{i+1}(x) = C$. This implies that there exists $wx \in E(G)$ such that $C_i(w) = C$ and $p_i(w) > p_i(x)$. Since $p_i(w) >0$, it follows that $N_C[w] = \Delta$. Therefore  $p_i(v) = \frac{1}{\Delta}$. Since $N_C[x] \neq \emptyset$ and $|N_D[x]| \geq 2$ it follows that $p_0(x) \geq \frac{1}{\Delta}$, a contradiction.
	
	If the  $\frac{\Delta-1}{\Delta+1}$-power index process on $G$ seeded with $C_0$ is eventually cyclic, then there exists a vertex $x$ and $i>j > 0$ such that $C_i(x) = C, C_{i+1}(x) = D$ and $C_j(x) = D$ and $C_{j+1}(x) = C$. 	
	If $C_i(x) = C$ and $C_{i+1}(x) = D$ then there exists a neighbour, $y$ of $x$ such that $p_i(y) > p_i(x)$ and $C_i(y) = D$. 
	Since $p_i(y) > 0$,  we have $|N_D[y]| \geq 2$. 
	Therefore $C_s(y) = D$ for all $s \geq i$. 
	However, by the previous statements, $C_t(x) = D$  for all $t \geq i+1$, as $xy\in E(G)$ and $C_{i+1}(x) = C_{i+1}(y) = D$. 
	This contradicts that the process is eventually periodic.
	
	\emph{Case II:  $w \in \alpha_{k}$} 
	
	By Theorem \ref{thm:PartitionWorks} we need only consider the $\frac{\Delta}{\Delta+1}$-power index process on $G$.  Let $C_0$ be a configuration and let $p_0$ be the power  of vertices of $G$ with respect to $C_0$. Consider $v$ such that $C_0(v) = C$. 
	If $w = \frac{\Delta}{\Delta+1}$, then $p_0(v) > 0$ if and only if $N_C[v] = N[v]$. 
	Therefore if $x$ and $y$ are adjacent vertices such that $C_0(x) = C$ and $C_0(y)=D$, then $C_1(x) = C_1(y) = D$.
	 And so if there exists $v$ such that $C_0(v) = D$ then $C_0$ is defector dominant after at most $diam(G)$ time-steps.
	
\end{proof}

\begin{corollary}
	If $G$ is a cycle, path or a $3$-regular graph and $w \in \left[\frac{1}{2},1\right)$, then every initial configuration is eventually stable in the $w$-power index process.
\end{corollary}

\begin{proof}
	If $G$ is a cycle, path or a $3$-regular graph then the win partition of $G$ contains no more than two parts. 
	The result now follows directly from Theorem  \ref{thm:LastTwo}.
\end{proof}

Though the $w$-power index process on cycles and $3$-regular graphs cannot result in a cycle for any initial configuration. The same is not true $k$-regular graphs for $k > 3$.

\begin{theorem}
	Let $\alpha_1, \alpha_2, \dots \alpha_j$ be the win partition of a $k$-regular graph. 
	For every $k > 3$ and every $w \notin \alpha_{k-1} \cup \alpha_k$ there exists a $k$-regular graph $G$, and an initial configuration $C_0$ on $G$ such that $C_0$ is eventually periodic.
\end{theorem}

\begin{proof}
	Let $G = K_{j-1} \square C_4$ and $w = \frac{1}{2}$. 
	Label the four induced copies of $K_{j-1}$ as  $G_1$, $G_2$, $G_3$, and $G_4$ so that vertices in $G_i$ are adjacent to those in $G_{i+1 \pmod 4}$. 
	Define $C_0$ so that each vertex of $G_1$ is a collaborator and all of other vertices are defectors. 
	If $w \notin \alpha_{k-1} \cup \alpha_k$, then $w < \frac{k-1}{k+1}$. 
	Therefore

	\[ p_{0}(v) = \begin{cases} 
	\frac{1}{j-1}, & v \in G_1\\
	\frac{1}{j}. & v \in G_2 \cup  G_4    \\
	\frac{1}{j+1}, & v \in G_3\\
	\end{cases} \] 
	
	Since $\frac{1}{j} < \frac{1}{j-1}$, each vertex of $G_2$ and $G_3$ will change their strategy to be a collaborator in $C_1$. 
	Therefore 			
	
	\[ p_{1}(v) = \begin{cases} 
	\frac{1}{j+1}, & v \in G_1\\
	\frac{1}{j}. & v \in G_2 \cup  G_4    \\
	\frac{1}{j-1}, & v \in G_3\\
	\end{cases} \] 
	
	It follows then that each vertex of $G_2$ and $G_3$ will change their strategy to be a collaborator in $C_2$. 
	The conclusion follows by observing that $C_0 = C_2$, and thus the process is periodic with period $2$.
\end{proof}

\section{Periodicity Results} \label{sec:stable}
In this section we highlight the full range of possibilities for the outcome of the $\frac{1}{2}$-power index process. 
In particular, we show that for $w = \frac{1}{2}$ there exist graphs and initial configurations that result in arbitrarily long cycles. 

Consider the graph $G_n = C_n \square P_2$, with vertex set $\{v_{i,j} \; : \; 1 \leq i \leq n, 1 \leq j \leq 2 \}$.
We obtain $H_n$ from $G_n$ by attaching a pendant edge with new vertex $z_{i,j}$ to each vertex $v_{i,j}$ of $G_n$, and identifying the vertex at the end of each pendant edge with a vertex in a copy of $K_3$ with vertices $x_{i,j},y_{i,j},z_{i,j}$. 
Figure \ref{8wave} gives $H_8$.

\begin{figure}
\[
\begin{scriptsize}
\begin{tikzpicture}
[scale = 1]
\tikzset{vertex/.style = {shape=circle,draw}}

\foreach \s in {0,...,7}
{
   \node[vertex] (\s0) at (\s,0) {};
   \node[vertex] (\s1) at (\s,1) {};
   
   \node[vertex] (t\s01) at (\s,-1) {};
   \node[vertex] (t\s02) at (\s - 0.2,-1.3) {};
   \node[vertex] (t\s03) at (\s + 0.2,-1.3) {};
   
   \node[vertex] (t\s11) at (\s,2) {};
   \node[vertex] (t\s12) at (\s - 0.2,2.3) {};
   \node[vertex] (t\s13) at (\s + 0.2,2.3) {};
   
   \draw (\s0) -- (\s1);
   
   \draw (t\s01) -- (t\s02) -- (t\s03) -- (t\s01);
   \draw (t\s11) -- (t\s12) -- (t\s13) -- (t\s11);
   
   \draw (\s0) -- (t\s01);
   \draw (\s1) -- (t\s11);
}

\draw (00) -- (10) -- (20) -- (30) -- (40) -- (50) -- (60) -- (70);
\draw (01) -- (11) -- (21) -- (31) -- (41) -- (51) -- (61) -- (71);

\draw (00) to[bend left=10] (70);
\draw (01) to[bend right=10] (71);

\end{tikzpicture}
\end{scriptsize}
\]
	\caption{The Graph $H_n$}
	\label{8wave}	
\end{figure}
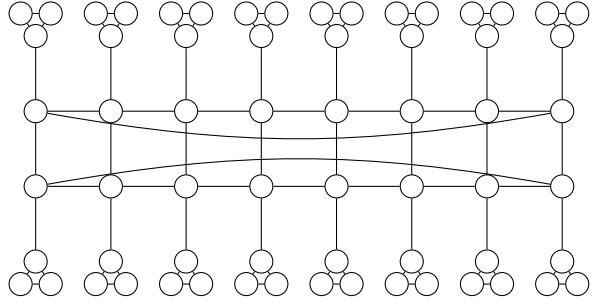

We define the configuration $W$, where for all $1 \leq i \leq n$
\begin{itemize}
	\item $W(v_{i,1})= W(x_{i,1}) = W(y_{i,1}) = W(z_{i,1}) = C$; and
	\item $W(v_{i,2})= W(x_{i,2}) = W(y_{i,2}) = W(z_{i,2}) = D$.
\end{itemize}
 Figure \ref{8wavecoloured} shows $W$ on $H_8$. 

\begin{figure}
\[
\begin{scriptsize}
\begin{tikzpicture}
[scale = 1]
\tikzset{vertex/.style = {shape=circle,draw}}

\foreach \s in {0,...,7}
{
   \node[vertex,fill=black] (\s0) at (\s,0) {};
   \node[vertex] (\s1) at (\s,1) {};
   
   \node[vertex,fill=black] (t\s01) at (\s,-1) {};
   \node[vertex,fill=black] (t\s02) at (\s - 0.2,-1.3) {};
   \node[vertex,fill=black] (t\s03) at (\s + 0.2,-1.3) {};
   
   \node[vertex] (t\s11) at (\s,2) {};
   \node[vertex] (t\s12) at (\s - 0.2,2.3) {};
   \node[vertex] (t\s13) at (\s + 0.2,2.3) {};
   
   \draw (\s0) -- (\s1);
   
   \draw (t\s01) -- (t\s02) -- (t\s03) -- (t\s01);
   \draw (t\s11) -- (t\s12) -- (t\s13) -- (t\s11);
   
   \draw (\s0) -- (t\s01);
   \draw (\s1) -- (t\s11);
}

\draw (00) -- (10) -- (20) -- (30) -- (40) -- (50) -- (60) -- (70);
\draw (01) -- (11) -- (21) -- (31) -- (41) -- (51) -- (61) -- (71);

\draw (00) to[bend left=10] (70);
\draw (01) to[bend right=10] (71);

\end{tikzpicture}
\end{scriptsize}
\]
	\caption{Configuration $W$ on $H_n$}
	\label{8wavecoloured}	
\end{figure}
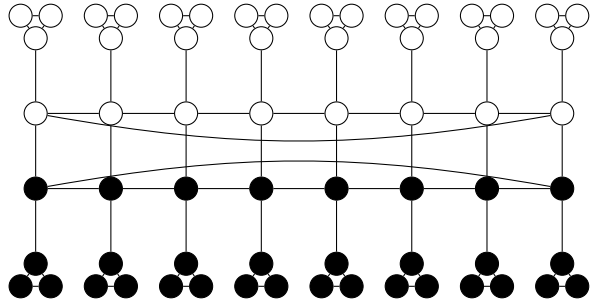

By inspection, such a configuration is stable in the $\frac{1}{2}$-power index process. 
We note that in the $\frac{1}{2}$-power index process on $H_n$, since each vertex is of degree three or degree four the outcome of the process is unchanged with respect to periodicity if the roles of collaborator and defector are swapped. 

From $W$ we define a family of configurations called \emph{wave configurations}. 
Let $n = 2k$ for some integer $k > 0$.
A \emph{wave configuration of length $n$} is a configuration of $H_n$ formed from $W$ by changing the strategies of any subset, $I$, of vertices from one of the copies of $C_n$, such that if  $v_{i_1,j}, v_{i_2,j} \in I$ then $|i_1 - i_2| \equiv 0 \mod 2$. 
We call such vertices \emph{interrupters}. 
Figure \ref{fig:waveConfigs} shows wave configurations on $H_8$.

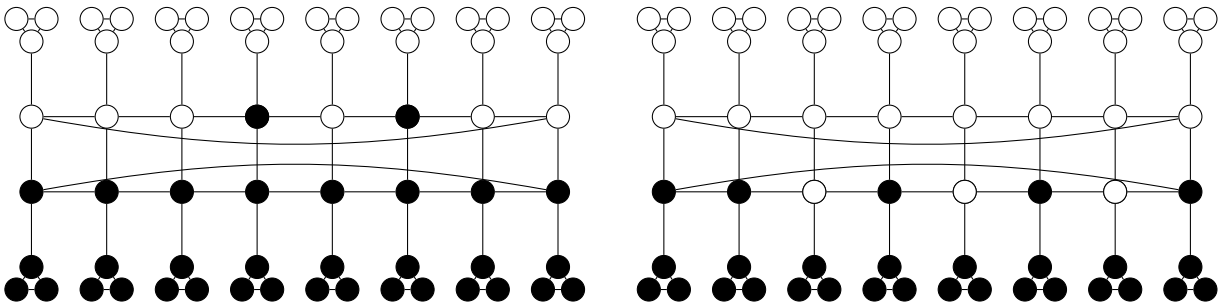
\begin{figure}
\[
\begin{scriptsize}
\begin{tikzpicture}
[scale = 1]
\tikzset{vertex/.style = {shape=circle,draw}}

\foreach \s in {0,...,7}
{
   \node[vertex,fill=black] (\s0) at (\s,0) {};
   \node[vertex] (\s1) at (\s,1) {};
   
   \node[vertex,fill=black] (t\s01) at (\s,-1) {};
   \node[vertex,fill=black] (t\s02) at (\s - 0.2,-1.3) {};
   \node[vertex,fill=black] (t\s03) at (\s + 0.2,-1.3) {};
   
   \node[vertex] (t\s11) at (\s,2) {};
   \node[vertex] (t\s12) at (\s - 0.2,2.3) {};
   \node[vertex] (t\s13) at (\s + 0.2,2.3) {};
   
   \draw (\s0) -- (\s1);
   
   \draw (t\s01) -- (t\s02) -- (t\s03) -- (t\s01);
   \draw (t\s11) -- (t\s12) -- (t\s13) -- (t\s11);
   
   \draw (\s0) -- (t\s01);
   \draw (\s1) -- (t\s11);
}

\node[vertex,fill=black] (31) at (3,1) {};\node[vertex,fill=black] (51) at (5,1) {};

\draw (00) -- (10) -- (20) -- (30) -- (40) -- (50) -- (60) -- (70);
\draw (01) -- (11) -- (21) -- (31) -- (41) -- (51) -- (61) -- (71);

\draw (00) to[bend left=10] (70);
\draw (01) to[bend right=10] (71);

\end{tikzpicture}
\end{scriptsize}
\ \ \ \ \
\begin{scriptsize}
\begin{tikzpicture}
[scale = 1]
\tikzset{vertex/.style = {shape=circle,draw}}

\foreach \s in {0,...,7}
{
   \node[vertex,fill=black] (\s0) at (\s,0) {};
   \node[vertex] (\s1) at (\s,1) {};
   
   \node[vertex,fill=black] (t\s01) at (\s,-1) {};
   \node[vertex,fill=black] (t\s02) at (\s - 0.2,-1.3) {};
   \node[vertex,fill=black] (t\s03) at (\s + 0.2,-1.3) {};
   
   \node[vertex] (t\s11) at (\s,2) {};
   \node[vertex] (t\s12) at (\s - 0.2,2.3) {};
   \node[vertex] (t\s13) at (\s + 0.2,2.3) {};
   
   \draw (\s0) -- (\s1);
   
   \draw (t\s01) -- (t\s02) -- (t\s03) -- (t\s01);
   \draw (t\s11) -- (t\s12) -- (t\s13) -- (t\s11);
   
   \draw (\s0) -- (t\s01);
   \draw (\s1) -- (t\s11);
}

\node[vertex,fill=white] (20) at (2,0) {};
\node[vertex,fill=white] (40) at (4,0) {};
\node[vertex,fill=white] (60) at (6,0) {};

\draw (00) -- (10) -- (20) -- (30) -- (40) -- (50) -- (60) -- (70);
\draw (01) -- (11) -- (21) -- (31) -- (41) -- (51) -- (61) -- (71);

\draw (00) to[bend left=10] (70);
\draw (01) to[bend right=10] (71);

\end{tikzpicture}
\end{scriptsize}
\]
	\caption{$C$ and $D$-wave configurations}
	\label{fig:waveConfigs}	
\end{figure}
If the interrupters in a wave configuration are collaborators, we will refer to this configuration as a $C$\emph{-wave configuration}. Analogously we define  $D$\emph{-wave configuration}. 
Note that for any wave configuration, if $u$ is an interrupter then $p(u) = 0$. 

\begin{lemma}\label{lem:CtoD}
	If $C_k$ is a $C$-wave configuration of length $n$ with interrupter set $I$, then in the $\frac{1}{2}$-power index process we have that $C_{k+1}$ is a $D$-wave configuration of length $n$ with interrupter set 
	
	\[ I^\prime = \{v_{1,j}\; | \; \mbox{ exactly one of } v_{j+1,2} \mbox{ and  } v_{j-1,2} \mbox{ is contained in } I, 1 \leq j \leq n
 \}\]
\end{lemma}

\begin{proof}
	Let $C_k$ be a $C$-wave configuration of length $n$.
	Note that in any $C$-wave configuration of length $n$, we have that $p(u) \leq \frac{1}{3}$ for all $u \in V(H_n)$.
	
	If $C_k = W$, the claim holds as $I = \emptyset$, $I^\prime = \emptyset$ and $C_{i+1} = W$.

	So we may assume there is at least one interrupter. 
	We examine the behaviour of each vertex $v$ by considering its position in $H_n$.
	
	\emph{Case I: $v$ is in a copy of  $K_3$}.
	If $v$ is adjacent to a vertex with opposing strategy, that neighbouring vertex must be an interrupter. Recall that if $u$ is an interrupter then $p(u) = 0$. And so $C_k(v) = C_{k+1}(v)$, as $v$ has no neighbours with opposing strategies and with strictly greater power.
	
	\emph{Case II: $v = v_{i,1}$}.
	By hypothesis, $C_k$ is a $C$-wave configuration.
	If $v$ is adjacent to an interrupter, then every vertex in $v$'s neighbourhood is a collaborator. 
	As such $C_k(v) = C_{k+1}(v)$.
	
	Assume now that $v$ is adjacent to no interrupter.
	If $u$ is a neighbour of $v$ and $C_k(v) = C$, then $p(u) \leq \frac{1}{4}$. 
	Note that this bound is achieved with equality if $u$ is contained a copy of $K_3$.
	If $u$ is a neighbour of $v$ and $C_k(v)=D$, then $v = v_{i,2}$.
	Observe that $p(v)$ depends on the number of interrupters that are adjacent to $u$.
	If $v$ has no neighbours that are interrupters, then each of its neighbours are defectors. 
	And so $p(v_{i,2}) = \frac{1}{4}$.
	In this case we see $C_k(v) = C_{k+1}(v)$.
	If $v$ has a single neighbour that is an interrupter, then $p(v_{i,2}) = \frac{1}{3}$.
	In this case we see $C_{k+1}(v) = D$.
	Finally, if $v$ has two neighbours that are interrupters, then $p(v_{i,2}) = 0$.
	In this case we see $C_k(v) = C_{k+1}(v)$.

	\emph{Case III: $v = v_{i,2}$}
	If $v$ is an interrupter, then $p(v) = 0$ and $C_{k+1}  = D$.
	Assume now that $v$ is not an interrupter. 
	Note that $v$ is adjacent to a vertex $u$ contained in a copy of $K_3$.
	By Case I, $p(u) = \frac{1}{3}$.
	Since $\frac{1}{3}$ is the maximum value taken by $p$ in $C_k$, we observe that $C_k(v) = C_{k+1}(v) = D$

	Therefore, the only vertices to change strategies are interrupters, and vertices $v_{i,1}$ where exactly one of $v_{i-1,2}$ and $v_{i+1,2}$ is an interrupter.
\end{proof}
	
\begin{lemma}
		If $C_i$ is a $D$ wave configuration of length $n$, then $C_{i+1}$ in the $\frac{1}{2}$-power index process is a $C$-wave configuration of length $n$.
\end{lemma}

\begin{proof}
	This proof follows similarly to the proof of Lemma \ref{lem:CtoD} by observing that in the $\frac{1}{2}$-power index process since each vertex of $H_n$ is of degree three or degree four the evolution of the process is unchanged by exchanging the roles of collaborator and defector.
\end{proof}

In the proof of Lemma \ref{lem:CtoD} we observe a set of rules that govern the generation of the subsequent wave configuration from the previous wave configuration. 
In particular, vertex $v_{i,1}$ is an interrupter at time $t+1$ if and only if exactly one of $v_{i-1,2}$ and $v_{i+1,2}$ is an interrupters at time $t$.

Recall that a 1D cylindrical cellular automaton of length $n$ is a deterministic discrete-time process consisting of $n$ cells: $c_0,c_1,\dots, c_{n-1}$.
At time $t$, each of the cells is either \emph{live} or \emph{dead}. 
The state of each cell $c_i$ at time $t= k+1$ depends only on the states of cells $c_{i-1},c_i$ and $c_{i+1}$ at time $t=k$. 
Here and further, we assume that addition and subtraction in subscripts of cell states is performed modulo $n$.
We call the state at $t=0$ the \emph{seed} of the cellular automaton.

Rule 90 is a 1D cylindrical cellular automaton of length $n$ whose evolution is  governed by the following rule: cell $c_i$ is live at time $t = k+1$ if and only if at most one of cell  $c_{i-1}$ and $c_{i+1}$ is live at time $t = k$ \cite{W84}.
By considering interrupters as live cells in a 1D cylindrical cellular automaton, we find that the evolution of a $C$-wave configuration may be modelled using Rule 90.

\begin{theorem}\label{thm:STTheorem}
Let $C_0$ be a C-wave configuration of length $n$ so that vertex $v_{0,2}$ is the only interrupter. Configuration $C_k$ in the $\frac{1}{2}$-power index process has an interrupter in column $i$ if and only if cell $i$ is live at time $t=k$ in the cylindrical cellular automaton of length $n$ governed by Rule 90 seeded with a single live cell in $c_{0}$.
\end{theorem}

\begin{proof}
	By Lemma \ref{lem:CtoD}, there is an interrupter in column $i$ in $C_k$ if and only if there is an interrupter in either column $i-1$ or column $i+1$ in $C_{k-1}$.
	There is an interrupter in column $i$ in $C_k$ if and only if  exactly one of cells $i-1$ and $i+1$ is live at time $t = k-1$ in  the cylindrical cellular automata of length $n$ governed by Rule 90 seeded with a single live cell.
\end{proof}

\begin{corollary}\label{cor:halfAllCycle}
	For every positive integer $k$ there exists a graph $G$ and an initial configuration $C_0$ such that the $\frac{1}{2}$-power index process eventually becomes periodic with period at least $k$.
\end{corollary}

\begin{proof}
One may confirm that Rule 90 of length $n = 2^k +2$ seeded with a single index turned on has period $2^{k-1}$.
The result follows by Theorem \ref{thm:STTheorem}.
\end{proof}

Corollary \ref{cor:halfAllCycle} may be extended so that $G$ contains as an induced subgraph any graph $S$.

\begin{corollary}\label{cor:halfAllCycleH}
	For every positive integer $k$ and any graph $S$, there exists a graph $G$ that has $S$ as an induced subgraph, and an initial configuration $C_0$ such that the $\frac{1}{2}$-power index process eventually becomes periodic with period at least $k$.
\end{corollary}

\begin{proof}
	Let $S$ be a graph and let $k$ be a positive integer.
	By Corollary \ref{cor:halfAllCycle}, there exists a graph $G^\star$ and an initial configuration $C_0$ such that the $\frac{1}{2}$-power index process eventually becomes periodic with period at least $k$.
	By the proof of Corollary \ref{cor:halfAllCycle} we may assume that $G^\star = H_n$ for some $n > 0$.
	Construct $G$ from $S \cup H_n$ by attaching any vertex $s$ in $S$ to any vertex $v$ of degree 2 in $H_n$. Now assign the usual configuration to $H_n$. Without loss of generality we may assume $C_0 (v) = C$. Now, assign $C_0(u) = C$ for all $u \in V(S)$. We must verify that $(a)$, $C_t(S) = C_0(S)$ for all $t$, and $(b)$, $H_n$ still behaves like the wave configuration. Notice that $(a)$ holds if $(b)$ does, so it is enough to show $(b)$. Observe the following subgraph, where $v$ and $s$ are defined above, and $i$ is an interrupter :
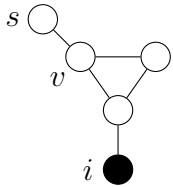
\[
\begin{tikzpicture}
[scale = 1]
\tikzset{vertex/.style = {shape=circle,draw}}

\node[vertex] (0) at (0,0) {};
\node (t0) at (-0.3,-0.3) {$v$};
\node[vertex] (1) at (1,0) {};
\node[vertex] (2) at (0.5,-0.71) {};

\node[vertex,fill=black] (3) at (0.5,-1.5) {};
\node (t3) at (0.1,-1.5) {$i$};

\node[vertex] (4) at (-0.5,0.5) {};
\node (t4) at (-0.9,0.5) {$s$};

\draw (0) -- (1) -- (2) -- (0);
\draw (2) -- (3);
\draw (0) -- (4);

\end{tikzpicture}
\]
Observe that the power  of $i$'s $C$ neighbour remains the same, meaning the interrupter will still change. 
Also, it's clear that none of the vertices other than $i$ will change. 

Therefore, $G$ contains $S$ as a subgraph, and we can find a configuration which is eventually periodic with arbitrarily long period. 
\end{proof}

The construction of $H_n$ can be extended to provide a result analogous to Corollary \ref{cor:halfAllCycleH} for the $w$-power index game for any $w \in \left[\frac{1}{2},1\right)$. 
For $\ell \geq  3$ let $H_{n,\ell}$ be constructed as $H_n$ where $K_3$ is replaced with $K_\ell$. 
Using $H_{n,\ell}$ one may prove the following result:

\begin{theorem}
	For every positive integer $k$ and any graph $S$, there exists a graph $G$ that has $S$ as an induced subgraph, and an initial configuration $C_0$ such that the $w$-power index process eventually becomes periodic with period at least $k$.
\end{theorem}

\section{Discussion}

Many open areas of study concerning the power index process on graphs remain. 
It is clear that for a fixed win condition, both the initial condition and the particular topology of the graph impact the limiting behaviour of the process. This is unsurprising, as it is a feature of many discrete-time processes on graphs, especially those with connections to cellular automata. 
That any induced subgraph may appear on a graph for which a configuration exists with arbitrarily long period length suggests that structural results beyond what is presented in Section \ref{sec:prelim} are unlikely.

Here the study of the process is motivated by the aim of discovering the full range of possible behaviours for the process.  
As the spread of attitudes is of particular interest on small-world networks these families of graphs are excellent choices in the continuing study of the power index process.

In many of the results presented herein, graphs and initial configurations are carefully crafted to achieve a desired property within in the process. 
The completely deterministic behaviour of the process makes such constructions possible. 
Recent work in the spread of influence of graphs with probabilistic dynamics (see \cite{C05}, \cite{D17}) gives a reasonable path for future work on a probabilistic version, including a probabilistic updating scheme and a random initial assignment of strategies. 

\bibliographystyle{abbrv}
\bibliography{references.bib}

\end{document}